\documentclass[conference]{IEEEtran}
\IEEEoverridecommandlockouts
\usepackage{xcolor,soul,framed} 
\colorlet{shadecolor}{yellow}
\usepackage[pdftex]{graphicx}
\graphicspath{{../pdf/}{../jpeg/}}
\DeclareGraphicsExtensions{.pdf,.jpeg,.png}
\usepackage{amsthm}

\usepackage[cmex10]{amsmath}
\usepackage{float}
\usepackage{xcolor,soul,framed} 
\usepackage{mathrsfs}
\usepackage{amsfonts}
\usepackage{bbm}
\usepackage{amssymb}
\usepackage{mwe,subcaption,tikz}
\usepackage{array}
\usepackage{mdwmath}
\usepackage{mdwtab}
\usepackage{eqparbox}
\usepackage{url}
\usepackage{cite}
\usepackage{amsmath}
\usepackage[noend]{algpseudocode}
\usepackage{algorithmicx,algorithm}
\usepackage{caption}
\usepackage{comment}
\DeclareMathOperator{\Tr}{Tr}
\DeclareMathOperator{\Diag}{Diag}
\DeclareMathOperator{\Argmin}{argmin}

\tikzset{boximg/.style={remember picture,red,thick,draw,inner sep=0pt,outer sep=0pt}}
\colorlet{shadecolor}{yellow}

\newtheorem{problem}{Problem}
\newtheorem{theorem}{Theorem}

\usepackage[bookmarks=false]{hyperref}

\usepackage[a-1b]{pdfx}

\def\BibTeX{{\rm B\kern-.05em{\sc i\kern-.025em b}\kern-.08em
    T\kern-.1667em\lower.7ex\hbox{E}\kern-.125emX}}

\begin{document}
\title{
Multi-Sensor Scheduling for Remote State Estimation over Wireless MIMO Fading Channels with Semantic Over-the-Air Aggregation
}

\author{\IEEEauthorblockN{Minjie~Tang$^*$,  Photios A. Stavrou$^*$, and Marios Kountouris${^*}^{\dagger}$}\\
$^*$Communication Systems Department, EURECOM, Sophia-Antipolis, France\\
$^{\dagger}$Department of Computer Science and
Artificial Intelligence, University of Granada, Spain\\
Emails: \texttt{\{minjie.tang, fotios.stavrou\}@eurecom.fr, mariosk@ugr.es}}

\maketitle
\thispagestyle{empty}
\pagestyle{empty}

\begin{abstract}
In this work, we study multi-sensor scheduling for remote state estimation over wireless multiple-input multiple-output (MIMO) fading channels using a novel semantic over-the-air (SemOTA) aggregation approach. We first revisit Kalman filtering with conventional over-the-air (OTA) aggregation and highlight its transmit power limitations. To balance power efficiency and estimation performance, we formulate the scheduling task as a finite-horizon dynamic programming (DP) problem. By analyzing the structure of the optimal $Q$-function, we show that the resulting scheduling policy exhibits a semantic structure that adapts online to the estimation error covariance and channel variations. To obtain a practical solution, we derive a tractable upper bound on the $Q$-function via a positive semidefinite (PSD) cone decomposition, which enables an efficient approximate scheduling policy and a low-complexity remote estimation algorithm. Numerical results confirm that the proposed scheme outperforms existing methods in both estimation accuracy and power efficiency.
\end{abstract}

\section{Introduction}

Remote state estimation is crucial in applications such as autonomous driving and industrial monitoring \cite{meng2023sensing, huang2020real, yang2025positioning}. A typical system consists of a potentially unstable \emph{dynamic plant}, distributed \emph{sensors}, and a \emph{remote estimator} (Fig.~\ref{fig:1}). The sensors monitor the plant states and transmit measurements over an unreliable wireless network, where the estimator processes the received (noisy) data to improve estimation accuracy. While wireless networks offer flexible deployment, constraints such as limited radio resources and transmit power can degrade estimation performance.

Remote state estimation over wireless networks has been extensively studied \cite{wu2019efficient, liu2022stability, chen2021resilient}, often assuming that each sensor is allocating dedicated radio resources. While this assumption simplifies communication, it results in low spectral efficiency, particularly as the number of sensors grows.
To mitigate this limitation, shared-resource access has been explored for geographically distributed sensors. Coordination mechanisms such as ALOHA \cite{yavascan2021analysis} and time-division multiple access (TDMA) \cite{liang2021data} enable multiple sensors to transmit data while managing interference. However, these schemes can suffer performance degradation due to packet collisions (in random access) and access latency (in scheduled access), which limits their suitability for real-time estimation.
Recently, over-the-air aggregation (OTA) has emerged as a promising alternative \cite{tang2020remote, tang2025csi}. By enabling simultaneous analog transmission of sensor measurements without explicit coordination overhead, OTA can substantially reduce access latency and improve spectral efficiency.
However, a major drawback is the high transmit power demand at the sensor nodes, which becomes increasingly critical as the number of active sensors grows. 

\begin{figure}
    \centering
    \includegraphics[height=2.5cm,width=5.5cm]{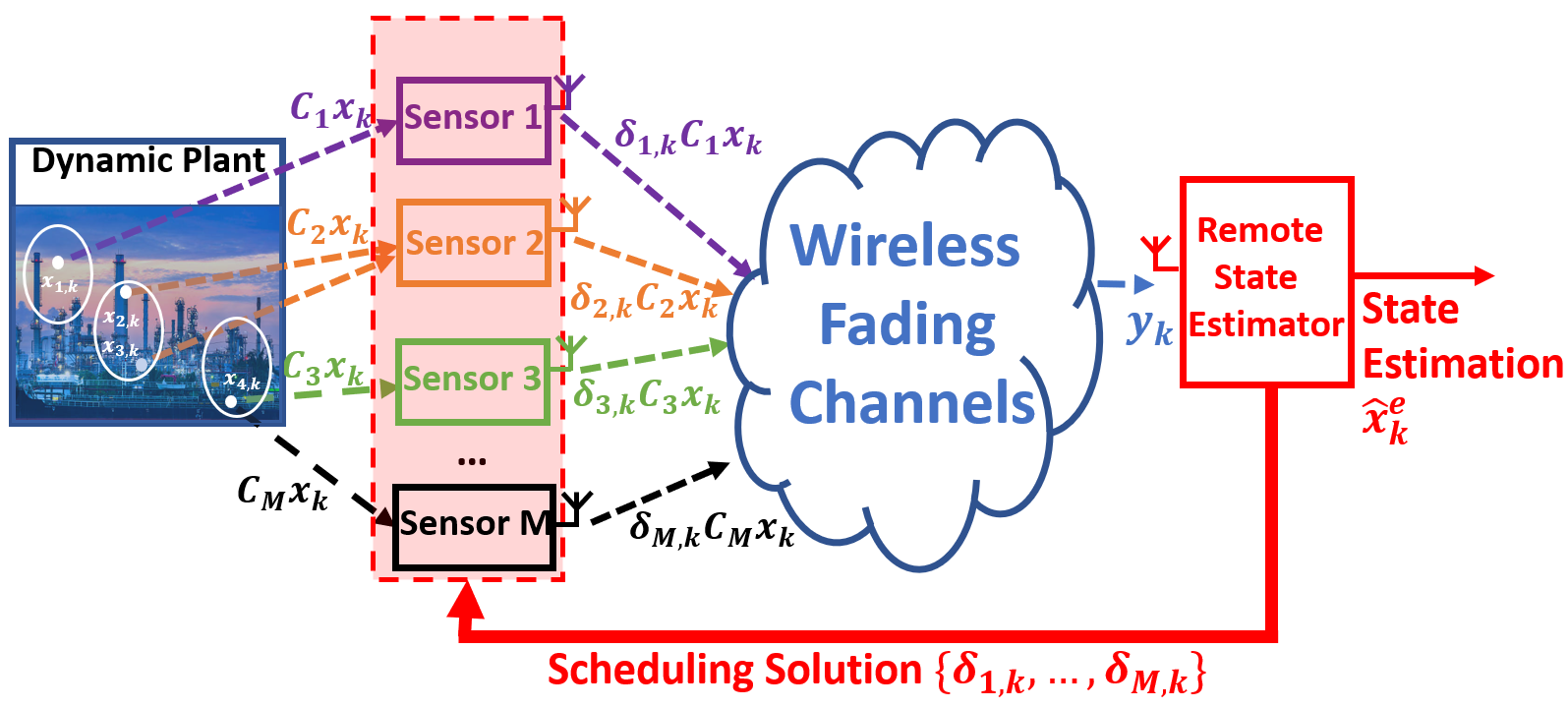}
    \caption{Architecture of a remote state estimation system.}
    \label{architecture}\label{fig:1}
\end{figure}
Recent advances in goal-oriented semantic communications \cite{kountouris2021semantics,stavrou:2023,stavrou2024indirect} focus on transmitting only application-relevant information, thereby reducing the energy consumption of communication devices. In remote state estimation, this paradigm prioritizes informative sensor measurements to improve estimation accuracy under energy constraints. Existing sensor scheduling strategies \cite{zhong2023event, zhong2024event, weerakkody2015multi} often use threshold-based triggering, where sensors transmit when measurements exceed predefined thresholds. However, such schemes can cause redundant transmissions, reduced estimation accuracy, and increased energy consumption, since measurement magnitude alone may not reflect the true informativeness of a measurement. Similarly, scheduling based solely on error covariance \cite{challagundla2020efficient, yang2023attack, han2015stochastic} ignores measurement heterogeneity, which can lead to inefficient power usage. Reinforcement learning (RL) has also been investigated for sensor scheduling \cite{yang2025structure} and often yields covariance-driven, threshold-type policies. Nonetheless, most prior work \cite{zhong2023event, zhong2024event, challagundla2020efficient, yang2023attack, weerakkody2015multi, han2015stochastic, yang2025structure} assumes dedicated radio channels, which limits spectral efficiency and oversimplifies practical communication dynamics. Moreover, when extended to random fading channels, these approaches may suffer further performance degradation in state estimation.

In this work, we propose a novel SemOTA-based sensor-scheduling framework for remote state estimation over random wireless MIMO fading channels. The main contributions are summarized as follows. {\bf (i)} We develop a SemOTA scheme tailored to MIMO fading channels. Unlike conventional OTA aggregation, which can impose a high transmit-power burden due to simultaneous sensor activation \cite{tang2020remote, tang2025csi}, the proposed approach improves energy efficiency by selectively prioritizing informative sensor measurements while maintaining accurate state estimation.
{\bf (ii)} We formulate sensor scheduling for remote state estimation as a finite-horizon stochastic dynamic programming (DP) problem \cite{bertsekas:2019}. Exploiting the structured evolution of the estimation error covariance, we characterize structural properties of the optimal $Q$-function and derive the associated scheduling policy under mild optimality conditions. {\bf (iii)} Using a positive semidefinite (PSD) cone decomposition of the $Q$-function, we obtain a tractable closed-form upper bound. Based on this bound, we develop an approximate scheduling policy and an efficient algorithm for remote state estimation over wireless MIMO fading channels.
{\bf (iv)} We validate the analysis through simulations, showing that the proposed method consistently outperforms several representative baselines in both estimation accuracy and power efficiency.
A key advantage of the proposed scheme in {\bf (i)} is its ability to improve state estimation accuracy while improving transmit-power utilization. Moreover, it preserves the zero-access-latency and high-spectral-efficiency benefits of conventional OTA schemes (see, e.g., \cite{tang2020remote, tang2025csi}). Furthermore, our DP formulation in {\bf (ii)} shows that the optimal scheduling policy exhibits a semantic structure, dynamically adapting to real-time variations in both the estimation error covariance and the channel state.

\emph{Notation:} Bold uppercase and lowercase letters denote matrices and vectors, respectively. $(\cdot)^T$ and $\Tr(\cdot)$ represent the transpose and trace. $\mathbf{0}_{m \times n}$ and $\mathbf{0}_m$ are zero matrices of size $m \times n$ and $m \times m$, respectively. $\mathbf{I}_S$ is the $S \times S$ identity matrix, and $\Diag(a, b, \ldots)$ forms a diagonal matrix.
Sets are denoted as follows: $\mathbb{R}^{m \times n}$ for real $m \times n$ matrices, $\mathbb{S}_+^m$ and $\mathbb{S}_{++}^m$ for positive semidefinite and positive definite $m \times m$ matrices, $\mathbb{Z}_+$ for positive integers, and $\mathbb{R}_+$ for positive real numbers.
Norms are given by $\|\mathbf{A}\|$ (spectral norm of $\mathbf{A}$) and $\|\mathbf{a}\|$ (Euclidean norm of $\mathbf{a}$). $[\mathbf{A}]_i$ denotes the $i$-th principal submatrix, while $[\mathbf{A}]_{a:b, c:d}$ extracts the submatrix from rows $a$ to $b$ and columns $c$ to $d$.
For matrices $\mathbf{A}, \mathbf{B}$ and vector $\mathbf{a}$: $\mathbf{A}(\cdot)^T = \mathbf{A}\mathbf{A}^T$, $\mathbf{A}\mathbf{B}(\cdot)^T = \mathbf{A}\mathbf{B}\mathbf{A}^T$, and $\mathbf{a}(\cdot)^T = \mathbf{a}\mathbf{a}^T$.
The indicator function $\mathbf{1}_{{\cdot}} \in \left\{0,1\right\}$ equals $1$ if the condition holds, otherwise it equals $0$. Expectation is denoted by $\mathbb{E}\left\{\cdot\right\}$.

\section{System Model}
\subsection{Dynamic Plant and Wireless Sensors}
We consider a time-slotted system where the dynamic plant is modeled by a first-order linear time-invariant difference equation for the state $\mathbf{x}_k \in \mathbb{R}^{S \times 1}$, given by:
\begin{align}
\mathbf{x}_{k+1}=\mathbf{A}\mathbf{x}_k+\mathbf{w}_k,
\end{align}
where $\mathbf{A}\in\mathbb{R}^{S\times S}$ defines the known plant dynamics, and $\mathbf{w}_k\sim\mathcal{N}(\mathbf{0}_{S\times 1}, \mathbf{W})$ represents zero-mean Gaussian process noise with a positive definite covariance matrix $\mathbf{W}\in\mathbb{S}_{++}^S$. The initial plant state is denoted by $\mathbf{x}_0=\mathbf{x}(0)\in\mathbb{R}^{S\times 1}$.

We consider $M \in \mathbb{Z}_+$ geographically distributed sensors monitoring the plant. Each sensor is equipped with $N_t \in \mathbb{Z}_+$ transmit antennas and shares a common radio resource pool for high spectral efficiency. The sensing model is given by:
\begin{align}
    \mathbf{z}_{m,k} = \mathbf{C}_m \mathbf{x}_k, \quad k = 0, \ldots, K-1, \quad m=1,\ldots,M,\nonumber
\end{align}
where \(\mathbf{C}_m \in \mathbb{R}^{N_t \times S}\) is the observation matrix of the \(m\)-th sensor, defining how it observes the plant state \(\mathbf{x}_k\). \(\mathbf{z}_{m,k}\in\mathbb{R}^{N_t\times 1}\) is the measurement collected by the \(m\)-th sensor.



\subsection{Wireless Communication Model}
The communication link between the sensors and the remote state estimator is modeled as a $N_r \times N_t$ wireless MIMO fading channel, where $N_r \in \mathbb{Z}_+$ denotes the number of receive antennas. Each active sensor transmits its measurement $\mathbf{z}_{m,k} \in \mathbb{R}^{N_t \times 1}$ via OTA aggregation \cite{tang2020remote, tang2025csi}. The received signal $\mathbf{y}_k \in \mathbb{R}^{N_r \times 1}$ is given by:
\begin{align}\label{receiving-model}
\mathbf{y}_k=\sum_{m=1}^M\delta_{m,k}\mathbf{H}_{m,k}\mathbf{z}_{m,k}+\mathbf{v}_k,~k=0,\ldots,K-1,
\end{align}
where $\mathbf{H}_{m,k} \in \mathbb{R}^{N_r \times N_t}$ represents the MIMO channel matrix, which remains constant within each timeslot and is independently and identically distributed (i.i.d.) across sensors and timeslots. Each entry of $\mathbf{H}_{m,k}$ follows a zero-mean Gaussian distribution with unit variance\footnote{Since  complex-valued MIMO channel can be  rewritten as an equivalent real-valued system by separating the real and imaginary parts, we adopt a real-valued channel model to simplify the notation.
}. The noise term $\mathbf{v}_k \sim \mathcal{N}(\mathbf{0}_{N_r \times 1}, \mathbf{I}_{N_r})$ represents additive white Gaussian noise (AWGN), and $\delta_{m,k} \in \{0,1\}$ is the sensor scheduling variable, detailed in Section~\ref{section:semota}.

\subsection{Remote State Estimator}
The objective of the remote state estimator is to compute the minimum mean-squared error (MMSE) estimate of the plant state $\mathbf{x}_k$ based on the received measurements using Kalman filtering.
Let $\mathbf{y}_0^k = \{\mathbf{y}_0, \dots, \mathbf{y}_k\}$ denote the history of received measurements up to timeslot $k$. We define 
$\widehat{\mathbf{x}}_k^e = \mathbb{E}[\mathbf{x}_k | \mathbf{y}_0^k]$, 
$\widehat{\mathbf{x}}_k = \mathbb{E}[\mathbf{x}_k | \mathbf{y}_0^{k-1}]$,  
$\Sigma_k^e = \mathbb{E}[(\mathbf{x}_k - \widehat{\mathbf{x}}_k^e)(\cdot)^T | \mathbf{y}_0^k]$  
and  
$\Sigma_k = \mathbb{E}[(\mathbf{x}_k - \widehat{\mathbf{x}}_k)(\cdot)^T | \mathbf{y}_0^{k-1}]$
where $\widehat{\mathbf{x}}_k^e \in \mathbb{R}^{S \times 1}$ and $\widehat{\mathbf{x}}_k \in \mathbb{R}^{S \times 1}$ represent the updated (posterior) and predicted (prior) plant states, respectively. Similarly, $\Sigma_k^e \in \mathbb{S}_+^S$ and $\Sigma_k \in \mathbb{S}_+^S$ denote the posterior and prior error covariance matrices of $\mathbf{x}_k$, respectively.
The Kalman filter recursively updates the state estimate as follows:
\begin{align}\label{prediction}
&\widehat{\mathbf{x}}_k=\mathbf{A}\widehat{\mathbf{x}}_{k-1}^e,\\&
\Sigma_k=\mathbf{A}\Sigma_{k-1}^e\mathbf{A}^T+\mathbf{W},\label{prediction-step-kalman-filter}
\\&\label{estimation}
\widehat{\mathbf{x}}_{k}^e=\widehat{\mathbf{x}}_{k}+\mathbf{K}_{k}(\mathbf{y}_k-\sum_{m=1}^M\delta_{m,k}\bar{\mathbf{H}}_{m,k}\widehat{\mathbf{x}}_k),
\\&\Sigma_k^e=(\mathbf{I}_S-\sum_{m=1}^M\delta_{m,k}\mathbf{K}_k\bar{\mathbf{H}}_{m,k})\Sigma_k(\cdot)^T+\mathbf{K}_k(\cdot)^T,\label{estimation-step-kalman-filter}
\end{align}
where  
$\bar{\mathbf{H}}_{m,k} = \mathbf{H}_{m,k} \mathbf{C}_m$, and 
$\mathbf{K}_k=\sum_{m=1}^M\delta_{m,k}\Sigma_k \bar{\mathbf{H}}_{m,k}^T((\sum_{m=1}^M\delta_{m,k}\bar{\mathbf{H}}_{m,k})\Sigma_k(\cdot)^T+\mathbf{I}_{N_r} )^{-1}\in \mathbb{R}^{S \times N_r}$ is the Kalman gain.

In conventional OTA aggregation schemes for remote state estimation, all sensors are active simultaneously and transmit their measurements over a shared radio resource, i.e., $\delta_{m,k}=1$ for all $m,k$. 
Since the aggregated noisy measurement $\sum_{m=1}^M \bar{\mathbf{H}}_{m,k}\mathbf{x}_k$ contains richer state information than individual measurements $\bar{\mathbf{H}}_{m,k}\mathbf{x}_k$, it improves estimation accuracy, as reflected by a lower expected trace of the error covariance $\mathbb{E}[\Tr(\Sigma_k)]$. 
However, this accuracy gain comes at the cost of higher transmit power consumption. 
Specifically, since sensor $m$ transmits $\mathbf{z}_{m,k}$, its transmit power satisfies
$\|\mathbf{z}_{m,k}\|^2 = \Tr(\mathbf{C}_m \mathbf{x}_k \mathbf{x}_k^T \mathbf{C}_m^T)$, which scales with the transmit power
$\Tr(\mathbf{C}_m \mathbf{C}_m^T)$ under the unit-variance normalization $\|\mathbf{x}_k\|^2 = 1$.
Accordingly, we define the transmit power cost of sensor $m$ as
$\Tr(\mathbf{C}_m \mathbf{C}_m^T)$, and the total transmission cost under conventional OTA is given by
$\sum_{m=1}^M \Tr(\mathbf{C}_m \mathbf{C}_m^T)$, which increases with the number of active sensors.
To reduce the power consumption in large-scale sensor networks, we propose a SemOTA-based sensor scheduling policy that balances power efficiency and estimation accuracy.

\section{Semantic Over-the-Air Aggregation}
\label{section:semota}
\subsection{Problem Formulation and Optimal Solution}
The prior error covariance $\Sigma_k$ evolves according to equations (\ref{prediction-step-kalman-filter}) and (\ref{estimation-step-kalman-filter}), and is expressed as:
\begin{align}
    \Sigma_{k+1}=f(\Sigma_k,\left\{\delta_{m,k}\right\},\left\{\mathbf{H}_{m,k}\right\})+\mathbf{W},\label{prediction-evolution}
\end{align}
where $f(\Sigma_k,\left\{\delta_{m,k}\right\},\left\{\mathbf{H}_{m,k}\right\})=\mathbf{A}(\Sigma_k^{-1}+(\sum_{m=1}^M\delta_{m,k}$ $\bar{\mathbf{H}}_{m,k})^T(\cdot))^{-1}\mathbf{A}^T.$

\begin{problem}\label{problem:1}(Problem Formulation for Sensor Scheduling) Consider the state estimation accuracy, measured by $\mathbb{E}[\Tr(\Sigma_k)]$, and the sensor power consumption, quantified by 
$\mathbb{E}[\sum_{m=1}^M \Tr(\delta_{m,k}\mathbf{C}_m \mathbf{C}_m^T)]$. 
The resulting sensor scheduling problem is formulated as a finite-horizon stochastic optimization over $K \in \mathbb{Z}_+$:
 \begin{align}
\min_{\pi}\mathbb{E}\left[\sum_{k=0}^{K-1}\Tr(\Sigma_k)+\sum_{m=1}^M\gamma \Tr(\delta_{m,k}\mathbf{C}_m\mathbf{C}_m^T)+\Tr(\Sigma_K)\right],\nonumber
\end{align}
where $\pi = \{\pi_k: k=0,1, \dots,{K-1}\}$ denotes the sensor scheduling policy over the horizon of $K$ time steps, with $\pi_k = \{\delta_{1,k}, \dots, \delta_{M,k}\}$ representing the sensor activation decisions at each time step $k$, and $\gamma$ is a trade-off parameter that balances estimation accuracy and power consumption.
\end{problem}

\subsection{Optimal Solution}
The optimal sensor scheduling solution, $\delta_{m,k}^*$, for Problem \ref{problem:1} can be formulated using DP recursions with $Q$-functions (also referred to as $Q$-factors) \cite{bertsekas:2019}, as summarized in the following theorem.

\begin{theorem}(DP Equation) Let $\mathbf{H}_{k}=\left\{\mathbf{H}_{1,k},...,\mathbf{H}_{M,k} \right\}$ denote the set of channel fading gains at time $k$.
For $k = K-1,\ldots, 0$, Problem \ref{problem:1} can be optimally solved backward in time using the following DP recursions based on $Q$-functions:
\begin{align}
\label{dp eq}
    &Q_k(\mathcal{M}_k) = \Tr(\Sigma_k) 
    + \sum_{m=1}^M \gamma \Tr(\delta_{m,k} \mathbf{C}_m \mathbf{C}_m^T)+ \nonumber\\
    &\mathbb{E} \bigg[\min_{\pi_{k+1}} Q_{k+1}(\mathcal{M}_{k+1}) 
    \big| \mathcal{M}_k \bigg],
\end{align}
where  
\begin{align}\label{Q-function-form-random-channel}
    &Q_k(\mathcal{M}_k) = \Tr(\Sigma_k) 
    + \sum_{m=1}^M \gamma \Tr(\delta_{m,k} \mathbf{C}_m \mathbf{C}_m^T) \nonumber\\
    &\quad + \Tr(f(\mathcal{M}_k)) 
    + (K-k) \Tr(\mathbf{W})  + \Delta_k(\mathcal{M}_k) \mathbf{1}_{k \leq K-2},
\end{align}
and
\begin{align}\label{delta-form-random-channel}
    &\Delta_k(\mathcal{M}_k) = 
    \mathbb{E} \bigg[ \min_{\pi_{k+1}} \bigg( 
    \sum_{m=1}^M \gamma \Tr(\delta_{m,k+1} \mathbf{C}_m \mathbf{C}_m^T) \nonumber\\
    &\quad + \Tr\big(f(f(\mathcal{M}_k) + \mathbf{W}, 
  \pi_{k+1}, \mathbf{H}_{k+1}\big)\big) \nonumber\\
    &\quad + \Delta_{k+1} \big(f(\mathcal{M}_k) + \mathbf{W}, 
\pi_{k+1},\mathbf{H}_{k+1}\big) \bigg) \bigg| \mathcal{M}_k  \bigg], \nonumber\\&\Delta_{K-1}(\mathcal{M}_{K-1})=0.
\end{align}
In the DP recursions given by Equation \eqref{dp eq}, the state consists of the estimation error covariance \( \Sigma_k \) and the channel state information (CSI) $\mathbf{H}_k$, while the control action corresponds to the sensor scheduling variable $\pi_k$. The state-action pair is given by \( \mathcal{M}_k = \left\{\Sigma_k, \mathbf{H}_{k}, \pi_k \right\} \). 
The expectations in Equations (\ref{dp eq}) and (\ref{delta-form-random-channel}) account for the randomness of \( \mathbf{H}_{k+1} \), averaging its effect conditioned on the current state-action pair \( \mathcal{M}_k \).

The optimal solution 
of \eqref{Q-function-form-random-channel} is given by $\delta_{m,k}^* = 1$ when
\begin{align}
&\Tr\big(f(\Sigma_{k},\{\delta_{-m}^*,\delta_{m,k} = 0, \mathbf{H}_{k}\})\big) \nonumber\\
&+ \Delta_k\big(\Sigma_{k},\{\delta_{-m}^*,\delta_{m,k} = 0, \mathbf{H}_{k}\}\big) \nonumber\\
&- \Tr\big(f(\Sigma_{k},\{\delta_{-m}^*,\delta_{m,k} = 1, \mathbf{H}_{k}\})\big) \nonumber\\
& - \Delta_k\big(\Sigma_{k},\{\delta_{-m}^*,\delta_{m,k} = 1, \mathbf{H}_{k}\}\big) \geq \Tr(\gamma\mathbf{C}_m\mathbf{C}_m^T), \label{optimal-solution}
\end{align}
where $\delta_{-m} = \{\delta_{i} \mid 1 \leq i \leq M, i \neq m\}$.    
\end{theorem}

\begin{proof}
See Appendix A.    
\end{proof} 

The sensor scheduling policy in Theorem~1 exhibits a semantic structure that adapts to the real-time error covariance $\Sigma_k$ and the CSI $\mathbf{H}_{k}$. Specifically, the left-hand side (LHS) of \eqref{optimal-solution} measures the marginal improvement in state estimation accuracy obtained from activating sensor $m$ (i.e., transmitting $\mathbf{z}_{m,k}$). Consequently, sensors are activated only when their transmissions provide sufficient estimation gains.


\subsection{Approximate Solution}
The optimal sensor scheduling solution in Theorem~1 depends on $\Delta_k(\mathcal{M}_k)$, which is computationally expensive due to its iterative form in \eqref{delta-form-random-channel}. To overcome this issue, we propose an efficient scheduling solution by exploiting an upper bound on $\Delta_k(\mathcal{M}_k)$ within the $Q$-function, $Q_k(\mathcal{M}_k)$. Specifically, by leveraging a positive semidefinite (PSD) cone decomposition \cite{tang2020remote} applied to the aggregated channel gain $\sum_{m=1}^M \delta_{m,k} \bar{\mathbf{H}}_{m,k} \in \mathbb{R}^{N_r \times S}$, we derive a theoretical result that provides an approximate sensor scheduling solution.

\begin{theorem} (Approximate Solution of \eqref{dp eq}) \label{ap solution}
Let $(\sum_{m=1}^M \delta_{m,k} \bar{\mathbf{H}}_{m,k})^T(\cdot) = \mathbf{U}_k^T \Psi_k (\cdot)$, 
where the diagonal elements of $\Psi_k$ are arranged in descending order.  
Define $\gamma_k$ as the rank of $\Psi_k$, and let $\Pi_k = \Diag(\mathbf{1}_{\gamma_k}, \mathbf{0}_{S-\gamma_k}) \in \mathbb{S}_{+}^S$. 
We denote by 
$\alpha(\left\{\pi_{k},\mathbf{H}_{k}\right\}) = \|(\mathbf{A} \mathbf{U}_k^T (\mathbf{I}_S - \Pi_k) \mathbf{U}_k)^T(\cdot)\|$,  $
    \beta(\left\{\pi_{k},\mathbf{H}_{k}\right\})= \|\mathbf{A}\|^2 \Tr([\Psi_k]_{\gamma_k}^{-1})$, and $\bar{\alpha}=\mathbb{E}[\alpha(\left\{\delta_{m,k}=1, \mathbf{H}_{k}\right\})]$.
Then, the approximate  sensor scheduling solution to Problem \ref{problem:1} is given by $\delta_{m,k}^a = 1$ if  
{\begin{align}
    & (1 + (\sum_{i=2}^{K-k} \bar{\alpha}^{i-1} \alpha(\{\delta_{-m}^a, \delta_{m,k} = 0, \mathbf{H}_{k} \}) ) 
    \mathbf{1}_{k\leq K-2}) \nonumber\\
    & \times \Tr\big(f(\Sigma_k, \{\delta_{-m}^a, \delta_{m,k} = 0, \mathbf{H}_{k} \})\big) \nonumber\\
    & + \sum_{i=1}^{K-k-1} \bar{\alpha}^i \beta(\{\delta_{-m}^a, \delta_{m,k} = 0, \mathbf{H}_{k} \}) 
    \mathbf{1}_{k\leq K-2} \nonumber\\
    & - ( 1 + (\sum_{i=2}^{K-k} \bar{\alpha}^{i-1} \alpha(\{\delta_{-m}^a, \delta_{m,k} = 1, \mathbf{H}_{k} \}) ) 
    \mathbf{1}_{k\leq K-2}) \nonumber\\
    & \times \Tr\big(f(\Sigma_k, \{\delta_{-m}^a, \delta_{m,k} = 1, \mathbf{H}_{k} \})\big) \nonumber\\
    & - \sum_{i=1}^{K-k-1} \bar{\alpha}^i \beta(\{\delta_{-m}^a, \delta_{m,k} = 1, \mathbf{H}_{k} \}) 
    \mathbf{1}_{k\leq K-2} \nonumber\\&\geq \Tr(\gamma \mathbf{C}_m \mathbf{C}_m^T),
\end{align}
}and $\delta_{m,k}^a = 0$ otherwise.
\end{theorem}
\begin{proof} See Appendix B.
\end{proof}


Building on Theorem~\ref{ap solution}, we propose a sensor scheduling algorithm for remote state estimation using SemOTA, as outlined below.\footnote{Steps~1 and~2 in Algorithm~\ref{algo:1} require the channel state information (CSI) \(\mathbf{H}_{m,k}\), which can be obtained via standard channel estimation methods (e.g., see \cite{tse2005fundamentals}).}

\begin{algorithm}
\small
\caption{SemOTA for Remote State Estimation}
\textbf{Initialization: } $\widehat{\mathbf{x}}_0\leftarrow \widehat{\mathbf{x}}(0)\in\mathbb{R}^{S\times 1}$, $\Sigma_0\leftarrow \Sigma(0)\in\mathbb{S}_{++}^S$.

\textbf{For $k=0,1,...,K-1$:}

\begin{itemize}

      \item \textbf{Step 1 (Update of Sensor Scheduling Solution): }
    \begin{itemize}
        \item $\left\{\delta_{1,k},...\delta_{M,k}\right\}\leftarrow$ Compute at the remote state estimator based on Theorem 2 using $\Sigma_{k}$ and $\mathbf{H}_k$;
        \item $\left\{\delta_{1,k},...\delta_{M,k}\right\}$ is broadcast to the sensors. 
    \end{itemize}

    \item  \textbf{Step 2 (SemOTA for State Estimation):}
    \begin{itemize}
        \item $\mathbf{y}_k\leftarrow$  based on (\ref{receiving-model}) using $\left\{\delta_{1,k},...\delta_{M,k}\right\}$ and $\left\{\mathbf{z}_{1,k},...,\mathbf{z}_{M,k}\right\}$
         \item $\widehat{\mathbf{x}}_k^e\leftarrow$ based on (\ref{estimation});
        \item $\Sigma_k^e\leftarrow$ based on (\ref{estimation-step-kalman-filter});
        \item $\widehat{\mathbf{x}}_{k+1}\leftarrow$ based on (\ref{prediction});
        \item $\Sigma_{k+1}\leftarrow$ based on (\ref{prediction-step-kalman-filter}).
       
    \end{itemize}
\end{itemize}

\textbf{End}

\label{algo:1}
\end{algorithm}

We conclude this section with a theorem that provides a performance guarantee for Algorithm \ref{algo:1}.
\begin{theorem} (Performance of {Algorithm \ref{algo:1}}) The total cost in Problem \ref{problem:1} under Algorithm \ref{algo:1} is given by  
{\begin{align}
    &\mathbb{E}^{\pi^a} \bigg[ \sum_{k=0}^{K-1} \Tr(\Sigma_k) 
    + \sum_{m=1}^M \gamma \Tr(\delta_{m,k} \mathbf{C}_m \mathbf{C}_m^T) 
    + \Tr(\Sigma_K) \bigg] \nonumber\\
    &\quad \leq \mathbb{E} \bigg[\min_{\pi_0} Q^{a}_0(\mathcal{M}_0) \bigg],
\end{align}
}where \(\mathbb{E}^{\pi^a}[\cdot]\) denotes the expectation under the sensor scheduling policy \(\pi^a = \{\pi^a_0, \pi^a_1, \ldots\}\). The policy \(\pi_k^a\) at timeslot $k$ is given by
\(
\pi^a_k = \{\delta_{1,k}^a, \delta_{2,k}^a, \ldots, \delta_{M,k}^a\},
\)
as defined in Step 1 of Algorithm \ref{algo:1}, and \(Q^a_k(\mathcal{M}_k)\) is defined in (\ref{Q-a}).
\end{theorem}
\begin{proof} Note that the sensor scheduling solution in Algorithm \ref{algo:1} is given by $\pi^a_k = \Argmin_{\pi_k} Q^a_k(\mathcal{M}_k)$. 
This implies that:
\begin{align}
    &\mathbb{E}^{\pi^a} \left[ \sum_{k=0}^{K-1} \Tr(\Sigma_k) + \sum_{m=1}^M \gamma \Tr(\delta_{m,k} \mathbf{C}_m \mathbf{C}_m^T) + \Tr(\Sigma_K) \right] \nonumber\\
    &= \mathbb{E}_{\mathbf{H}} \bigg[ \min_{\pi_0} Q_0^a(\mathcal{M}_0) - \min_{\pi_1} Q_1^a(\mathcal{M}_1) + \ldots + \Tr(\Sigma_K) \bigg] \nonumber\\
    &= \mathbb{E}_{\mathbf{H}} \left[ \min_{\pi_0} Q_0^a(\mathcal{M}_0) \right],
\end{align}
where $\mathbb{E}_{\mathbf{H}}\left\{\right\}$ denotes the expectation over $\mathbf{H}=\left\{\mathbf{H}_0,...,\mathbf{H}_{K-1}\right\}$.
This completes the proof.
\end{proof}


\section{Numerical Results}
We evaluate the performance of SemOTA by comparing it against several baseline schemes.
\begin{itemize}
    \item \textbf{Baseline 1:} \emph{(Measurement Threshold with ALOHA)} Sensor $m$ is active if $\|\mathbf{z}_{m,k}\| \geq \sigma_1$ for a given threshold $\sigma_1 > 0$. When multiple sensors are active, their transmissions are coordinated using ALOHA.
    \item \textbf{Baseline 2:} \emph{(Error Covariance Threshold with Random TDMA)} When $\|\Sigma_k\| \geq \sigma_2$ for a given threshold $\sigma_2 > 0$, a single sensor is randomly selected to transmit with probability $\frac{1}{M}$ using random TDMA.
    \item \textbf{Baseline 3:} \emph{(OTA)} All sensor measurements are transmitted simultaneously via OTA.
\end{itemize}

The system dynamics are modeled as:
\begin{equation}
\mathbf{x}_{k+1} =
\begin{bmatrix}
 1.04 & 0.03 & 0.01 \\ 
 0.22 & 0.48 & 0.03 \\ 
 0.021 & 0.004 & 0.78 
\end{bmatrix} 
\mathbf{x}_k + \mathbf{w}_k,
\end{equation}
where $\mathbf{w}_k \sim \mathcal{N}(\mathbf{0}_{3\times 1}, \mathbf{I}_3)$ denotes process noise. We set $N_t = N_r = 2$, and $\mathbf{C}_m \in \mathbb{R}^{2\times 3}$ are randomly generated with elements drawn from $\mathcal{N}(0, 1)$. We set $\gamma=0.4$ and $K=1000$. For fairness,  $\sigma_1$ and $\sigma_2$  are selected via a grid search over $[1,10]$ (step $0.5$) to minimize the average normalized mean square error (NMSE) over a 1000-step horizon.

\subsection{State Estimation NMSE versus the Number of Sensors}
Fig.~\ref{sensor-1} shows the state estimation NMSE as a function of the number of sensors.
Compared to Baseline~1 and Baseline~2, which suffer from severe performance degradation as the number of sensors increases due to collisions and access latency bottlenecks, the proposed scheme maintains robust performance and achieves more than a two-order-of-magnitude NMSE reduction.
This is because OTA aggregation improves estimation accuracy by coherently combining more state information.
On the other hand, the proposed scheme exhibits slightly higher NMSE than Baseline~3, while remaining within the same order of magnitude, since it selectively transmits only the most critical sensor measurements, resulting in some information loss.



\begin{figure}[ht]
    \centering
    \begin{subfigure}[t]{0.49\linewidth}
        \centering
        \includegraphics[height=2.5cm,width=4.2cm]{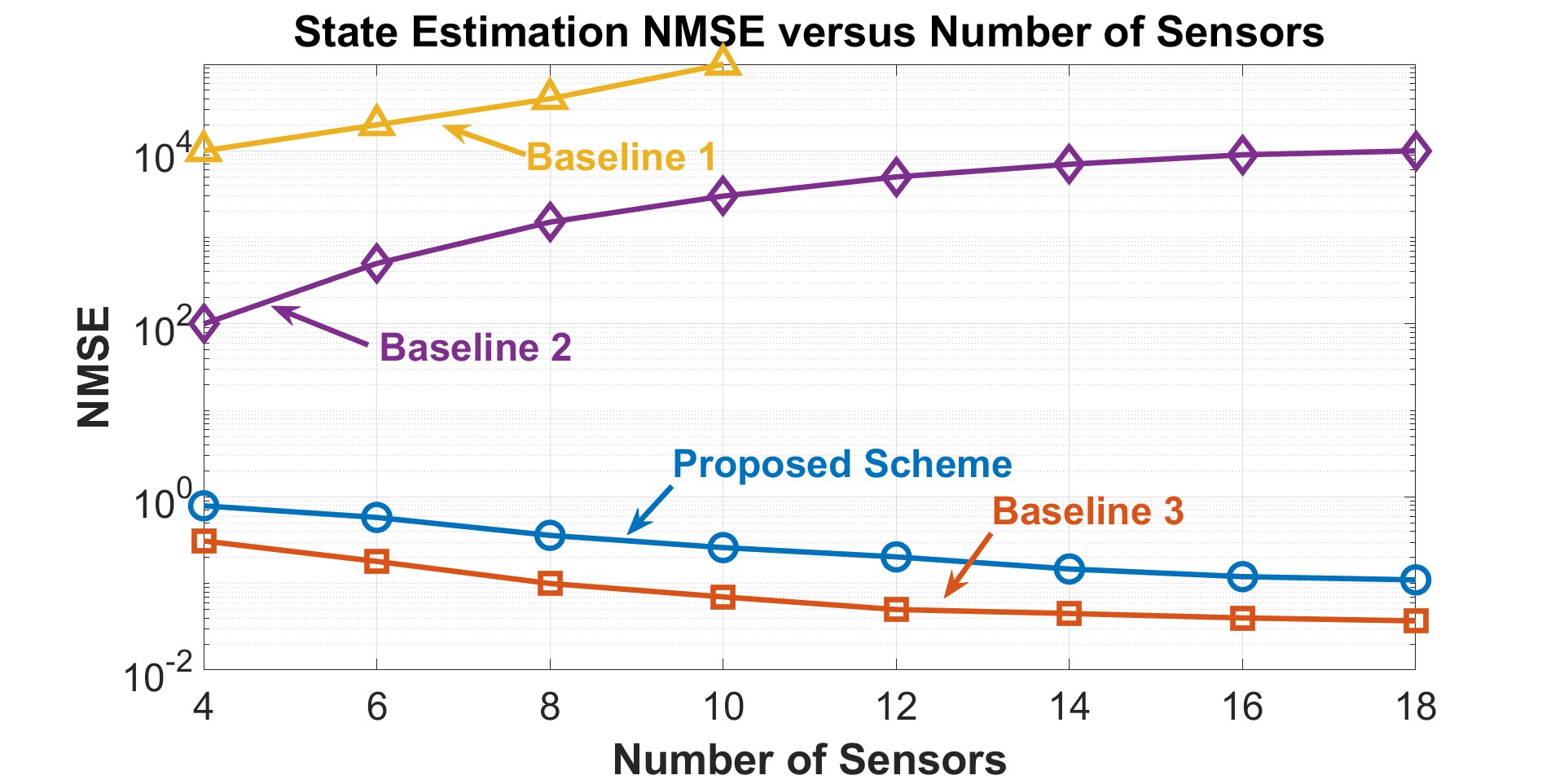}
        \caption{State estimation NMSE versus the number of sensors $M$.}
        \label{sensor-1}
    \end{subfigure}
    \hfill
    \begin{subfigure}[t]{0.49\linewidth}
        \centering
        \includegraphics[height=2.5cm,width=4.2cm]{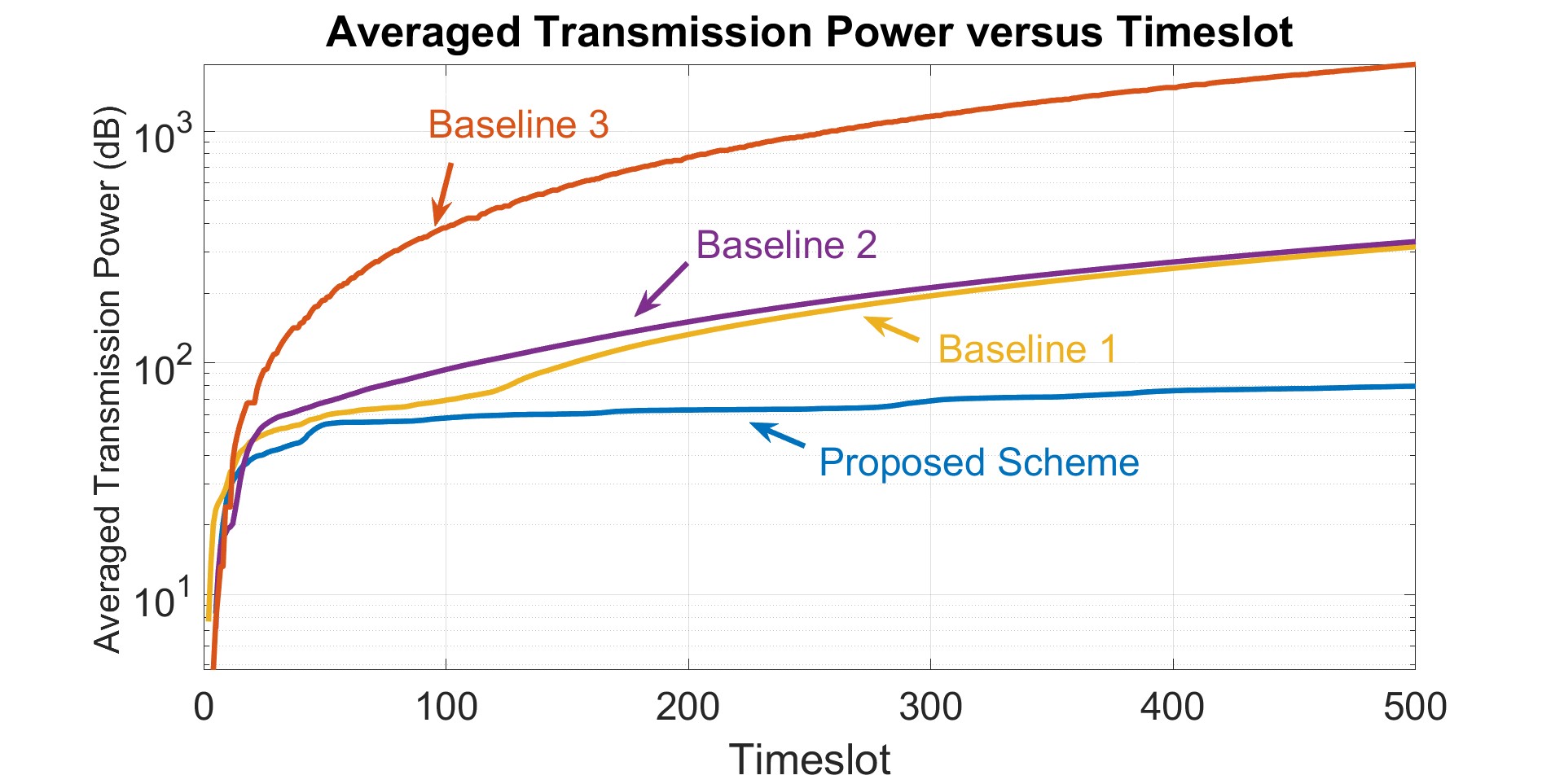}
        \caption{Average transmit power at the sensors versus timeslot for $M=8$.}
        \label{sensor-2}
    \end{subfigure}
    \caption{Numerical results.}
    \label{sensor-fig}
\end{figure}

\subsection{Power Consumption at Sensors versus Timeslot}

Fig.~\ref{sensor-2} illustrates the average transmission power at the sensors,
$\sum_{m=1}^M \mathbb{E}[\|\delta_{m,k}\mathbf{z}_{m,k}\|^2]$, over time.
Compared with the baseline schemes, the proposed approach operates at least one order of magnitude lower in transmit power throughout the entire horizon,
since it selectively transmits only the most critical sensor measurements.

\section{Conclusions}
We have investigated multi-sensor scheduling for remote state estimation over wireless MIMO fading channels using a novel SemOTA-based aggregation approach. By formulating the problem as a finite-horizon dynamic programming (DP) problem, we analyzed the optimal $Q$-function and the corresponding scheduling policy. Our results reveal a semantic-aware structure that dynamically adapts to real-time error covariance and channel variations. Using positive semidefinite (PSD) cone decomposition, we derived an upper bound on the $Q$-function, enabling an approximate solution and an efficient scheduling algorithm. Numerical results demonstrate that the proposed scheme outperforms existing approaches in both state estimation accuracy and power efficiency, confirming its practical effectiveness.
\appendix

\subsection{Proof of Theorem 1}
The proof is constructed using backward induction. 
Specifically, at the $(K-1)$-th step, given the terminal cost $\Tr(\Sigma_K)$, we obtain:
\begin{align}
    &Q_{K-1}(\mathcal{M}_{K-1}) = \Tr(\Sigma_{K-1}) + \sum_{m=1}^M \gamma \Tr(\delta_{m,K-1} \mathbf{C}_m \mathbf{C}_m^T) \nonumber\\
    &\quad + \Tr(\mathbf{W}) + \Tr(f(\mathcal{M}_{K-1})).
\end{align}
The optimal solution is given by   $\delta_{m,K-1}^* = 1$ only if
\begin{align}
    &\gamma \Tr(\mathbf{C}_m \mathbf{C}_m^T) 
    + \Tr(f(\Sigma_{K-1}, \{\delta_{-m,K-1}^*, \delta_{m,K-1} = 1, 
    \mathbf{H}_{K-1} \})) 
    \nonumber\\&\leq \Tr(f(\Sigma_{K-1}, \{\delta_{-m,K-1}^*, \delta_{m,K-1} = 0, \mathbf{H}_{K-1} \})).
\end{align}

Similarly, at the $(K-2)$-th step, we have:  
\begin{align}
    &Q_{K-2}(\mathcal{M}_{K-2}) 
    = \Tr(\Sigma_{K-2}) + \sum_{m=1}^M \gamma \Tr(\delta_{m,K-2} \mathbf{C}_m \mathbf{C}_m^T) 
    \nonumber\\&+ \Tr(f(\mathcal{M}_{K-2}) + \Delta_{K-2}(\mathcal{M}_{K-2}) 
   + 2 \Tr(\mathbf{W}).
\end{align}
The optimal solution is given by  $\delta_{m,K-2}^*=1$ only if
\begin{align}
    &\Tr(\gamma \mathbf{C}_m \mathbf{C}_m^T)  
    + \Tr(f(\Sigma_{K-2}, \{\delta_{-m,K-2}^*, \delta_{m,K-2} = 1,\nonumber\\&  \mathbf{H}_{K-2} \}))  + \Delta_{K-2}(\Sigma_{K-2}, \{\delta_{-m,K-2}^*, \delta_{m,K-2} = 1,\mathbf{H}_{K-2} \})\nonumber\\&  \leq \Tr(f(\Sigma_{K-2}, \{\delta_{-m,K-2}^*, \delta_{m,K-2} = 0, \mathbf{H}_{K-2} \}))  +\nonumber\\&  \Delta_{K-2}(\Sigma_{K-2}, \{\delta_{-m,K-2}^*, \delta_{m,K-2} = 0, \mathbf{H}_{K-2} \}).
\end{align}

By backward induction, 
we can obtain that for any $k =K-1,\ldots,1,0$, the $Q$-function  
$Q_k(\mathcal{M}_k)$ is given by (\ref{Q-function-form-random-channel}).  
The optimal solution satisfies $\delta_{m,k}^* = 1$ if (\ref{optimal-solution}) holds; otherwise, $\delta_{m,k}^* = 0$. This completes the proof of Theorem 1.

\subsection{Proof of Theorem 2}

Denote \( \bar{\beta} = \mathbb{E}[\beta(\{\delta_{m,k}=1, \mathbf{H}_{k} \})] \), \( \sum_{m=1}^{M} \bar{\mathbf{H}}_{m,k}^T(\cdot) = \bar{\mathbf{U}}_k^T \bar{\Psi}_k(\cdot) \), where the diagonal elements of \( \bar{\Psi}_k \) are arranged in descending order. Further, let \( \text{rank}(\bar{\Phi}_k) = \bar{\gamma}_k \) and $\bar{\Pi}_k=\Diag(\mathbf{I}_{\bar{\gamma}_k},\mathbf{0}_{S-\bar{\gamma}_k})$. We decompose $\Sigma_{k}$ as $\Sigma_{k}=\Sigma_{k}^o+\Sigma_{k}^u$, where
$\Sigma_{k}^{o}=  \bar{\mathbf{U}}_{k}^T\left[\begin{array}{cc}
(\Sigma_k)_{\bar{\gamma}_k} & (\Sigma_k)_{\bar{\gamma}_k}\bar{\mathbf{L}}_{k}\\
\bar{\mathbf{L}}_{k}^{T}(\Sigma_k)_{\bar{\gamma}_k} & \bar{\mathbf{L}}_{k}^T(\Sigma_k)_{\bar{\gamma}_k}\bar{\mathbf{L}}_{k}
\end{array}\right]\bar{\mathbf{U}}_{k}$, $\Sigma_k^u=\bar{\mathbf{U}}_k^T(\mathbf{I}_S-\bar{\Pi}_k)\bar{\mathbf{U}}_k\Sigma_k(\cdot)^T-\bar{\mathbf{U}}_k^T\Diag(\mathbf{0}_{\bar{\gamma}_k},\bar{\mathbf{L}}_k^T(\Sigma_k)_{\bar{\gamma}_k}\bar{\mathbf{L}}_k)\bar{\mathbf{U}}_k$.  $(\Sigma_k)_{\bar{\gamma}_k}=(\bar{\mathbf{U}}_k\Sigma_k\bar{\mathbf{U}}_k^T)_{\bar{\gamma}_k}$, and 
$\bar{\mathbf{L}}_k=(\Sigma_k)_{\bar{\gamma}_k}^{-1}(\bar{\mathbf{U}}_k\Sigma_{k}\bar{\mathbf{U}}_k^T)_{1:\bar{\gamma}_k;\bar{\gamma}_{k}+1:S}$.
Note that $\bar{\Psi}_k\Sigma_k^u=\mathbf{0}_S$ and $\text{ker}(\bar{\Psi}_k\Sigma_k^o)=\text{ker}(\Sigma_k^o)$, where $\text{ker}(\cdot)$ denotes the matrix kernel.  We have:
{
\begin{align}\label{psd decomposition}
    &\mathbb{E}_{\mathbf{H}_{K-1}}[\Tr(f(\Sigma_{K-1}, 
    \{\delta_{m,K-1} = 1, \mathbf{H}_{m,K-1}\}))]
    \nonumber\\
    &=\mathbb{E}_{\mathbf{H}_{K-1}}[\Tr(\mathbf{A}\Sigma_{K-1}^{u}\mathbf{A}^T+\mathbf{A}((\Sigma_{K-1}^o)^{-1}+\sum_{m=1}^M\bar{\mathbf{H}}_{m,K-1}^T\nonumber\\&\bar{\mathbf{H}}_{m,K-1})^{-1}\mathbf{A}^T)]\nonumber\\&=
    \mathbb{E}_{\mathbf{H}_{K-1}}[\Tr(\mathbf{A}\bar{\mathbf{U}}_{K-1}^T(\mathbf{I}_S-\bar{\Pi}_{K-1})\bar{\mathbf{U}}_{K-1}\Sigma_{K-1}(\cdot)^T)]+\nonumber\\&\mathbb{E}_{\mathbf{H}_{K-1}}[\Tr(\mathbf{A}\bar{\mathbf{U}}_{K-1}^T\Tilde{\Psi}_{K-1}^T\nonumber\\
    &\small\begin{bmatrix}(\mathbf{I}_{\bar{\gamma}_{K-1}}+(\Tilde{\Sigma}_{K-1})_{\bar{\gamma}_{K-1}}^{-1})^{-1}&*\\\bar{\mathbf{L}}_{K-1}^T(\Sigma_{K-1})_{\bar{\gamma}_{K-1}}( \mathbf{I}_{\bar{\gamma}_{K-1}}+ (\Tilde{\Sigma}_{K-1})_{{\bar{\gamma}}_{K-1}}^{-1})^{-1}&\eta_{K-1}\end{bmatrix}(\cdot)^T)]\nonumber
    \end{align}
    \begin{align}
    &= \mathbb{E}_{\mathbf{H}_{K-1}}[\Tr(\mathbf{A}\bar{\mathbf{U}}_{K-1}^T(\mathbf{I}_S-\bar{\Pi}_{K-1})\bar{\mathbf{U}}_{K-1}\Sigma_{K-1}(\cdot)^T)]\nonumber\\&+\mathbb{E}_{\mathbf{H}_{K-1}}[\Tr(\mathbf{A}\bar{\mathbf{U}}_{K-1}^T\nonumber\\&\!\!\!\!\!\!\!\!\small\begin{bmatrix}(\bar{\Psi}_{K-1})_{\bar{\gamma}_{K-1}}^{-\frac{1}{2}}(\mathbf{I}_{\bar{\gamma}_{K-1}}+(\tilde{\Sigma}_{K-1})_{\bar{\gamma}_{K-1}}^{-1})^{-1}(\bar{\Psi}_{K-1})_{\bar{\gamma}_{K-1}}^{-\frac{1}{2}}&*\\\bar{\mathbf{L}}_{K-1}^T(\Sigma_{K-1})_{\bar{\gamma}_{K-1}}(\mathbf{I}_{\bar{\gamma}_{K-1}}+(\widetilde{\Sigma}_{K-1})_{\bar{\gamma}_{K-1}}^{-1})^{-1}&-\widetilde{\eta}_{K-1}\end{bmatrix}\nonumber\\&(\cdot)^T)]
    \nonumber\\&\leq \bar{\beta} + \bar{\alpha} \Tr(\Sigma_{K-1}),
\end{align}
}where $(\Tilde{\Sigma}_k)_{\bar{\gamma}_k}=(\bar{\Psi}_k)^{\frac{1}{2}}_{\bar{\gamma}_k}(\Sigma_k)_{\bar{\gamma}_k}(\bar{\Psi}_k)^{\frac{1}{2}}_{\bar{\gamma}_k}$,
$\widetilde{\mathbf{P}}_{\gamma_k}=(\bar{\Psi}_k)^{\frac{1}{2}}_{\bar{\gamma}_k}(\Sigma_k)_{\bar{\gamma}_k}(\bar{\Psi}_k)^{\frac{1}{2}}_{\bar{\gamma}_k}$,   $\widetilde{\eta}_k=\bar{\mathbf{L}}_k^T(\Sigma_k)_{\bar{\gamma}_k}(\mathbf{I}_S+(\widetilde{\Sigma}_k)_{\bar{\gamma}_k}^{-1})^{-1}(\cdot)^T$, and $\eta_k=\bar{\mathbf{L}}_k^T(\Sigma_k)_{\bar{\gamma}_k}\bar{\mathbf{L}}_k-\bar{\mathbf{L}}_k^T(\Sigma_k)_{\bar{\gamma}_k}(\mathbf{I}_{\bar{\gamma}_k}+(\Tilde{\Sigma}_k)_{\bar{\gamma}_k})^{-1}(\Sigma_k)_{\bar{\gamma}_k}\bar{\mathbf{L}}_k.$
$*$ denotes the transpose of the value at position with row and column swapped.

This further gives that 
\begin{align}
    &\mathbb{E}_{\mathbf{H}_{K-1}} [ \Tr (f(f(\mathcal{M}_{K-2}) + \mathbf{W}, \{\delta_{m,K-1} = 1\}) ) ] \nonumber\\
    &\leq \bar{\alpha} \alpha(\{\pi_{K-2}, \mathbf{H}_{K-2}\}) \Tr(\Sigma_{K-2}) + \bar{\alpha} \Tr(\mathbf{W}) + \nonumber\\
    &\quad \bar{\alpha}  \beta(\{\pi_{K-2}, \mathbf{H}_{K-2}\}) + \bar{\beta}.
\end{align}

By backward induction, we can obtain that for any $k=K-2,...,1,0$,
$\Delta_k(\mathcal{M}_k)$ is upper bounded as follows:
{
\begin{align}\label{delta-upper-bound}
    &\Delta_k(\mathcal{M}_k) \leq \sum_{m=1}^M (K-k-1) \gamma \Tr(\mathbf{C}_m \mathbf{C}_m^T)  
     \nonumber\\
    &\!\!\! + \sum_{i=1}^{K-k-1} (K-k-i) \bar{\alpha}^i \Tr(\mathbf{W})+ (\sum_{i=2}^{K-k} \bar{\alpha}^{i-1} )  
    \alpha(\left\{\pi_k,\mathbf{H}_{k}\right\}) \nonumber\\
    &\times \Tr(f(\mathcal{M}_k))  + \sum_{i=1}^{K-k-1} \bar{\alpha}^i \beta(\pi_k,\mathbf{H}_{k}) 
    \nonumber\\&+ \sum_{i=1}^{K-k-2} (K-k-i-1) \bar{\alpha}^i \bar{\beta}  
    + (K-k-1) \bar{\beta}.
\end{align}
}By substituting $\Delta_k(\mathcal{M}_k)$ into the $Q$-function $Q_k(\mathcal{M}_k)$ using its upper bound (\ref{delta-upper-bound}), we obtain $ Q_k(\mathcal{M}_k)\leq  Q_k^a(\mathcal{M}_k),$
where $Q_k^a(\mathcal{M}_k)$ is defined as:
\begin{align}\label{Q-a}
    &Q_k^a(\mathcal{M}_k)\triangleq\Tr(\Sigma_k)+\sum_{m=1}^M\gamma\Tr(\delta_{m,k}^a\mathbf{C}_m\mathbf{C}_m^T)+\nonumber\\&\Tr(f(\mathcal{M}_k))+(K-k)\Tr(\mathbf{W})+\sum_{m=1}^M(K-k-1)\gamma \Tr(\mathbf{C}_m\nonumber\\&\times \mathbf{C}_m^T)+(K-k-1)\bar{\beta}+(\sum_{i=2}^{K-k}\bar{\alpha}^{i-1}\alpha(\pi_k,\mathbf{H}_{k})\nonumber\\&\Tr(f(\mathcal{M}_k))+\sum_{i=1}^{K-k-1}(K-k-i)\bar{\alpha}^i\Tr(\mathbf{W})+\sum_{i=1}^{K-k-1}\nonumber\\&\bar{\alpha}^i\beta(\pi_k,\mathbf{H}_{k})+\sum_{i=1}^{K-k-2}(K-k-i-1)\bar{\alpha}^i\bar{\beta})\mathbf{1}_{k\leq K-2}.
\end{align}
By minimizing $Q^a_k(\mathcal{M}_k)$ over the scheduling variable $\{\delta_{m,k}^a\}$, we obtain the structured form of $\{\delta_{m,k}^a\}$. This completes the proof.

\section*{Acknowledgment}
This work is supported by the SNS JU project 6G-GOALS \cite{strinati:2024} under the EU’s Horizon programme Grant Agreement No. 101139232. The work of P. A. Stavrou is also supported by the
Huawei France-EURECOM Chair on Future Wireless Networks.

\bibliographystyle{IEEEtran}
\bibliography{IEEEabrv,Bibliography}
\end{document}